\documentclass[letterpaper, 10 pt, conference]{ieeeconf}  % Comment this line out if you need a4paper
\IEEEoverridecommandlockouts                              % This command is only needed if 
\pdfoutput=1
\newcommand{\rulesep}{\unskip\ \vrule\ }
\newcommand{\rulehorzsep}{\unskip\ \hrule\ }
\usepackage{amsmath,amsthm,amssymb}
\usepackage{mathtools}
\mathtoolsset{showonlyrefs} % only equations that are referenced will be numbered
\DeclareMathOperator*{\argmin}{\arg\!\min}
\usepackage{amsfonts}
\usepackage{csquotes}
\usepackage{tikz}
\usepackage{graphicx, subfig, wrapfig}
\usepackage{algorithm}
\usepackage{algpseudocode}
\newtheorem{thm}{Theorem}

\newtheorem{lem}{Lemma}
\newtheorem{prop}{Proposition}
\usepackage{verbatim}
\usepackage{url}
\title{\LARGE \bf
	Safe Open-Loop Strategies for Handling Intermittent Communications in Multi-Robot Systems*
}

\author{Siddharth Mayya$^{1}$ and Magnus Egerstedt$^{2}$% <-this % stops a space
\thanks{*This research was sponsored by Grant No. 1544332 from the U.S. National Science Foundation}% <-this % stops a space
\thanks{$^{1}$Siddharth Mayya is with the School of Electrical and Computer Engineering, Georgia Institute of Technology, Atlanta, GA 30332, USA
        {\tt\small siddharth.mayya@gatech.edu}}%
\thanks{$^{2}$Magnus Egerstedt is a professor at the School of Electrical and Computer Engineering, Georgia Institute of Technology,
        Atlanta, GA 30332, USA
        {\tt\small magnus@gatech.edu}}%
}
\begin{document}
\maketitle
\thispagestyle{empty}
\pagestyle{empty}
%%%%%%%%%%%%%%%%%%%%%%%%%%%%%%%%%%%%%%%%%%%%%%%%%%%%%%%%%%%%%%%%%%%%%%%%%%%%%%%%
\begin{abstract}
In multi-robot systems where a central decision maker is specifying the movement of each individual robot, a communication failure can severely impair the performance of the system. This paper develops a motion strategy that allows robots to safely handle critical communication failures for such multi-robot architectures. For each robot, the proposed algorithm computes a time horizon over which collisions with other robots are guaranteed not to occur. These safe time horizons are included in the commands being transmitted to the individual robots. In the event of a communication failure, the robots execute the last received velocity commands for the corresponding safe time horizons leading to a provably safe open-loop motion strategy. The resulting algorithm is computationally effective and is agnostic to the task that the robots are performing. The efficacy of the strategy is verified in simulation as well as on a team of differential-drive mobile robots.
\end{abstract}

%%%%%%%%%%%%%%%%%%%%%%%%%%%%%%%%%%%%%%%%%%%%%%%%%%%%%%%%%%%%%%%%%%%%%%%%%%%%%%%%
\section{Introduction} \label{sec_intro}
Multi-robot systems have reached a point of maturity where they are beginning to be deployed in real-world scenarios (see e.g., the special issue \cite{kitts2008design} and constituent papers \cite{antonelli2008null,sariel2008naval}). Such deployments of robot teams often require a signal exchange network, typically in the form of a wireless communication channel. Communication is not only essential for sharing sensor measurements and performing diagnostics, it is often an integral part of the closed loop control mechanism \cite{balch1994communication,klavins2004communication}. In fact, in many applications and scenarios such as extra-terrestrial exploration \cite{yliniemi2014multirobot}, high precision manufacturing \cite{dogar2015multi}, and multi-robot testbeds \cite{pickem2016safe}, the robots frequently rely on communicating with a centralized decision maker for their velocity or position commands. In such situations, a failure in the communication network can severely hinder the motion of the robots and performance of the algorithm. This is the premise behind the work performed in this paper, whereby the adverse effects caused by intermittently failing communication networks are mitigated. \par
As deployment conditions for multi-robot systems become increasingly harsh, occasional failures in wireless communication channels are both expected and inevitable \cite{ramanathan2002brief}.
This raises the following question: {\it{What should a robot do in case a communication failure prevents it from receiving critical motion commands from a central decision maker?}} \par
Many different techniques have been explored to handle communication uncertainties in multi-robot teams \cite{arkin2002line,nguyen2003autonomous,ulam2004good}. In many cases, the robots are assumed to have significant decision making capability, have sensors to maneuver around obstacles, or have knowledge about the positions of other robots. Furthermore, some developed communication recovery techniques do not provide formal collision-avoidance guarantees in case of unforeseen communication failures. \par 
An existing technique used to handle communication failures, mentioned in \cite{ulam2004good}, is to stop the robot when critical data is not received. While this behavior preserves safety, it could cause the robot to behave erratically. For example, if only intermittent velocity commands are received, the robot could move in a jerky ``start-stop" fashion. This problem was observed on the Robotarium: a remote-access multi-robot testbed being developed at Georgia Tech \cite{pickem2016safe}. The Robotarium is a multi-robot research platform, and gives users the flexibility to test any coordination algorithm they wish. During experiments, it was observed that failures in the communication channels prevented robots from receiving velocity or position commands which caused them to abruptly stop moving. This lead to a disruption in the coordination algorithm being executed, and affected the ability of the Robotarium to faithfully reproduce the behavior specified by the user. \par
Motivated by the need to alleviate such problems in general, and resolve issues with the Robotarium in particular, this paper proposes a strategy that allows differential-drive robots without sensory or decision-making capabilities, to continue moving safely for a specific amount of time even when velocity commands from a central decision maker are not received. For each robot, the central decision maker computes a time horizon over which collisions with other robots are guaranteed not to occur. This is called the safe time horizon. During normal operations, the desired velocity and the safe time horizon are transmitted to the robots periodically. If a robot stops receiving data due to a communication failure, it executes the last received velocity command for the duration of the last received safe time horizon. This allows the robot to continue moving in a provably collision-free manner despite having no updated information about the environment. The robot can follow this open-loop trajectory for the entire duration of the safe time horizon, beyond which it stops moving.  \par
In order to calculate the safe time horizon, we first compute the set of all possible locations that can be reached by a robot within a given time (i.e., the reachable set \cite{fedotov2011reachable,soueres1994set}). This is followed by computing the time horizon for which each robot lies outside the reachable set of other robots. But, for the differential-drive robots considered here, performing such set-membership tests is computationally expensive owing to the non-convexity of the reachable set. Consequently, the reachable set is over-approximated by enclosing it within an ellipse whose convex structure allows for simpler set-membership tests and finite representation \cite{maler2008computing}. By minimizing the area of the ellipse enclosing the convex hull of the reachable set, we obtain the best ellipsoidal over-approximation of the reachable set in terms of the accuracy and effectiveness of set-membership tests. \par 
The outline of this paper is as follows: In Section \ref{sec_reachability_analysis}, the structure of the reachable sets of differential-drive robots is outlined. Section \ref{sec_approx_set} derives an ellipsoidal approximation of the reachable set. Section \ref{sec_safe_time_horizon} formally defines the safe time horizon, outlines the algorithm used by the robots, and proves the safety guarantees that it provides. In Section \ref{sec_results}, the developed algorithm is first implemented in simulation, following which, experimental verification is performed on a team of differential-drive mobile robots. Finally, Section \ref{sec_conclusion} concludes the paper.   
\section{Reachability Analysis} \label{sec_reachability_analysis}
Consider a dynamical system whose state evolves according to the following differential equation,
\begin{align} \label{general_non_linear_dynamics}
\dot {\mathbf{x}} = f(\mathbf{x},\mathbf{u}) , \;\; \mathbf{y} = h(\mathbf{x}), \;\; \mathbf{x}(0) = \mathbf{x_0},\;\; \mathbf{u} \in \mathcal{U}, \nonumber
\end{align}
where $\mathbf{x}$ is the state of the system, $\mathbf{y}$ represents the output of the system, $\mathbf{u}$ is the control input, and $\mathcal{U}$ is the set of admissible control inputs. The reachable set of outputs at time $t$ can be defined as 
\begin{equation*}
\mathcal{R}(t;\mathbf{x_0}) = \bigcup\limits_{\mathbf{u}(\cdot) \in \mathcal{U}} \mathcal{Y}(t,\mathbf{x_0},\mathbf{u}(\cdot)),
\end{equation*}
where $\mathcal{Y}(t,\mathbf{x_0},\mathbf{u}(\cdot))$ represents the output of the dynamical system at time $t$ under control action $\mathbf{u}(\cdot)$. \par
Since the Robotarium \cite{pickem2016safe}, in its current form, is populated with differential-drive mobile robots, this paper investigates the reachability of two-wheeled differential-drive robots with non-holonomic dynamics. \par 
For such systems, let $z = (x,y) \in \mathbb{R}^2$ denote the position of the robot in the 2D plane, and let $\phi \in [-\pi,\pi]$ denote its orientation with respect to the horizontal axis. As such, the robot is described as a point $(z,\phi)$ in the configuration space $\mathbb{R}^2 \times S^1$. The motion of the robots can be captured using the unicycle dynamics model:
\begin{align} \label{unicycle_dynamics_normalized}
& \dot x = v\cos(\phi), \; \dot y = v\sin(\phi), \; \dot \phi = \omega, \\
& |v| \leq 1, \; |\omega|\leq 1. \nonumber
\end{align}
The bounds on the linear velocity $v$ and the angular velocity $\omega$ represent the physical limitations of the robots, and are normalized to $1$ without loss of generality. \par 
Since collisions are ultimately defined by positions rather than orientations, we consider the output of the system to be $z$. Relative to this output, the reachable set $\mathcal{R}(t;z_0,\phi_0)$ is the set of all positions in the 2D plane that can be reached at time $t$ by a robot starting in the configuration $(z_0,\phi_0)$ at $t=0$. For $z_0 = (0,0), \phi_0 = 0$, denote the reachable set as $\mathcal{R}(t)$. Since the dynamics in \eqref{unicycle_dynamics_normalized} is drift-free, the structure of the reachable set does not depend on $z_0$ and $\phi_0$, i.e.,
\begin{equation}
\mathcal{R}(t;z_0,\phi_0) = z_0 + \Pi_{\phi_0}\mathcal{R}(t),
\end{equation}
where $\Pi_{\phi_0}$ is a rotation by $\phi_0$. Therefore, the study of reachable sets can be restricted to the case when $z_0 = (0,0)$ and $\phi_0 = 0$. \par
In order to reach the boundary of the reachable set, a robot must travel in a time-optimal manner \cite{boissonnat1994accessibility}. If this was not true, a point even farther away would be reachable in the same amount of time. A closely related motion model for which the structure of time-optimal paths has been extensively studied is the {\it{Reeds-Shepp car}} \cite{reeds1990optimal}. \par
The motion model of a {\it{Reeds-Shepp car}} is similar to \eqref{unicycle_dynamics_normalized}, except that the set of admissible inputs is $|v| = 1, |\omega| \leq 1$. In \cite{sussmann1991shortest}, the authors compute a family of trajectories rich enough to contain a time-optimal path between any two configurations for a robot with dynamics given in \eqref{unicycle_dynamics_normalized}. It is shown that every trajectory in this family is also time-optimal for a {\it{Reeds-Shepp car}}, and furthermore, this family of trajectories is sufficient for time-optimality of a {\it{Reeds-Shepp car}}. Another way to view this result is that, when moving in a time-optimal manner, the robot behaves like a {\it{Reeds-Shepp car}}. As mentioned earlier, any path reaching the boundary of the reachable set has to be time-optimal. Therefore, the reachable set for the {\it{Reeds-Shepp car}} and for the differential-drive robot considered here are identical. By utilizing expressions for the reachable set of a {\it{Reeds-Shepp car}} given in \cite{soueres1994set}, Fig. \ref{fig_reachSet} portrays the reachable set for a robot with dynamics given in \eqref{unicycle_dynamics_normalized}. \par
\begin{figure}[h]
	\centering
	\includegraphics[width=0.3\textwidth]{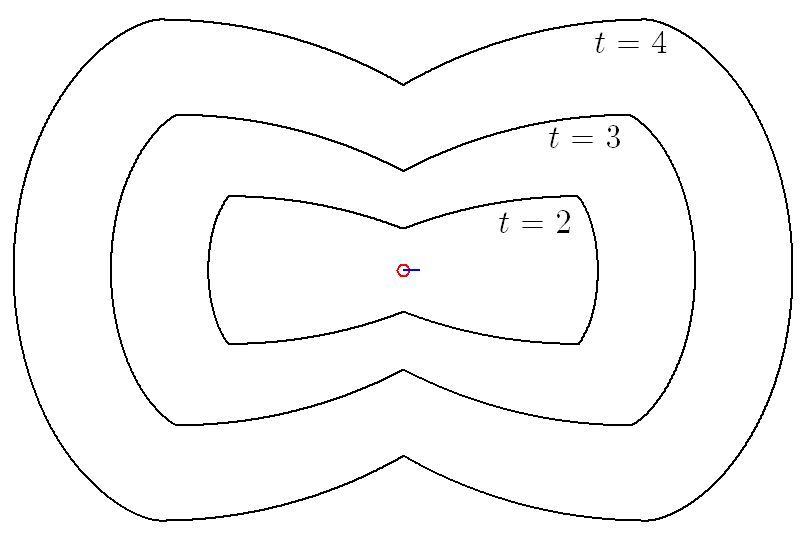}
	\caption{The reachable set $\mathcal{R}(t)$ of a robot with dynamics given in \eqref{unicycle_dynamics_normalized}, is depicted for varying time horizons. The robot is represented as a circle at the center.}
	\label{fig_reachSet}
\end{figure}
The complexity of performing set-membership tests with respect to the non-convex set $\mathcal{R}(t)$ increases the computational burden associated with the safe time horizon algorithm. Consequently, the next section derives an ellipsoidal approximation of the reachable set. 
\section{Approximating the Reachable Set} \label{sec_approx_set}
Given any convex set $K \subset \mathbb{R}^n$, there exists a unique ellipsoid of minimum volume circumscribing it \cite{john2014extremum}. This ellipsoid is denoted as $\xi(K)$. One way to derive an ellipsoidal approximation of the reachable set would be to compute the minimum area ellipse $\xi(conv(\mathcal{R}(t)))$, where $conv(\mathcal{R}(t))$ denotes the convex hull of $\mathcal{R}(t)$. Additionally, the derivation of analytical expressions for the ellipse $\xi(conv(\mathcal{R}(t)))$, if at all possible, will enable its efficient and fast computation in the safe time horizon algorithm. \par 

The feasibility of the derivation of analytical expressions for the ellipses, is closely linked to the symmetry properties of the underlying set as well as the equations describing it \cite{guler2012symmetry}. Consequently, in order to allow for analytical solutions, we introduce a new set $\mathcal{K}(t)$ enclosing $conv(\mathcal{R}(t))$, which allows for an easier computation of $\xi(\mathcal{K}(t))$. Essentially, this enables us to swap the problem of computing $\xi(conv(\mathcal{R}(t)))$ with the simpler problem of computing $\xi(\mathcal{K}(t))$. \par 
Furthermore, this approximation is justified by showing that the dissimilarity between $conv(\mathcal{R}(t))$ and $\mathcal{K}(t)$, as measured by the Jaccard distance metric \cite{jaccard1901distribution}, asymptotically goes to zero over time. For the sets $X, Y \subset \mathbb{R}^2$, the Jaccard distance based on the area measure is given by, 
\begin{equation}
d_J(X,Y) = 1 - \frac{\mathbf{A}(X \cap Y)}{\mathbf{A}(X \cup Y)},
\end{equation}
where $\mathbf{A}(\cdot)$ denotes the area of the set. The following proposition formally introduces the set $\mathcal{K}(t)$.
\begin{prop} \label{prop_simple_R}
	Let $\mathcal{K}(t)$ be given by,
	\begin{align}
	\mathcal{K}(t) = & \Bigg\{ p = \begin{bmatrix} p_x \\ p_y \end{bmatrix} \in \mathbb{R}^2 : \|p\|_2 \leq t, \vphantom{\begin{cases} 1-\cos(t) \\ t-\pi/2+1 \end{cases}} \\
	&
	\;\;\;|p_{y}| \leq \left \{
	\begin{array}{@{}ll@{}}
	1-\cos(t), \;\;\;\; \text{if} \;\;\; 0 < t \leq \pi/2 \\ t-\pi/2+1,\; \text{if} \;\;\; t>\pi/2
	\end{array} \right.
	\vphantom{\begin{cases} 1-\cos(t) \\ t-\pi/2+1 \end{cases}}
	\Bigg\}.
	\end{align}
	Then, $conv(\mathcal{R}(t)) \subset \mathcal{K}(t)$ and 
	\begin{equation}
	\lim_{t \to \infty} d_J(conv(\mathcal{R}(t)),\mathcal{K}(t)) = 0,
	\end{equation}
	where $d_J$ is the Jaccard distance.
\end{prop}
\begin{proof}
	Denote $\mathcal{C}^{++}(t)$ as the curved outer boundary of $conv(\mathcal{R}(t))$ in the first quadrant (see Fig. \ref{fig_approxRt}). From \cite{soueres1994set}, we know that, $\mathcal{C}^{++}(t)$ can be expressed as,
	\begin{align} %\label{eqn_reach_set_curves_c}
	& \mathcal{C}^{++}(t) = \Bigg\{\ p \in \mathbb{R}^2_{++}: \left\{
	\begin{array}{@{}ll@{}}
	p_x = \sin \psi + \gamma\cos \psi \\
	p_y = -\cos \psi + \gamma\sin \psi + 1
	\end{array} \right. \Bigg\}, \\
	& \text{where} \;\; \psi \in [0,\min(t,\pi/2)], \text{ and } \gamma = t-\psi.
	\end{align}
	In order to prove that $conv(\mathcal{R}(t)) \subset \mathcal{K}(t)$, we first show that all points on $\mathcal{C}^{++}(t)$ are closer to the origin than points on the outer boundary of $\mathcal{K}(t)$, which is a circular arc of radius $t$.
	Let $p(t,\psi) = (p_x(t,\psi),p_y(t,\psi)) $ denote a point on the curve $\mathcal{C}^{++}(t)$. The squared $l_2$-norm of $p(t,\psi)$ is given as,
	\begin{align*} %\label{norm_of_involute_curve}
	\|p(t,\psi)\|_2^2 &= (t-\psi)^2 + 2 - 2\cos(\psi) + 2(t-\psi)\sin(\psi),
	\end{align*}
	where $\psi \in [0,\min(t,\pi/2)]$. 
	The derivative of $\|p(t,\psi)\|_2^2$ with respect to the parameter $\psi$, is shown to be always non-positive:
	\begin{align} \label{circ_approx_proof1}
	\frac {\partial{\|p(t,\psi)\|_2^2}}{\partial{\psi}} &= 2(t-\psi)(\cos(\psi)-1) \leq 0 , \\ \forall  \; \psi &\in [0,\min(t,\pi/2)].
	\end{align}
	Since $p(t,0) = (t,0)$, this point lies on the circular arc of radius $t$. Then, the fact that $\frac {\partial{\|p(t,\psi)\|_2^2}}{\partial{\psi}} \leq 0$, implies that, as $\psi$ increases, the curve either overlaps with the circular arc of radius $t$, or gets closer to the origin. This proves that, a part of the outer boundary of $conv(\mathcal{R}(t))$, represented by the curve $\mathcal{C}^{++}(t)$, can be over-approximated by a circular arc of radius $t$. (see Fig. \ref{fig_approxRt}).
	\begin{figure}[h]
		\centering
		\includegraphics[trim = {0.1cm 0cm 0cm 0.0cm}, clip, width=0.4\textwidth]{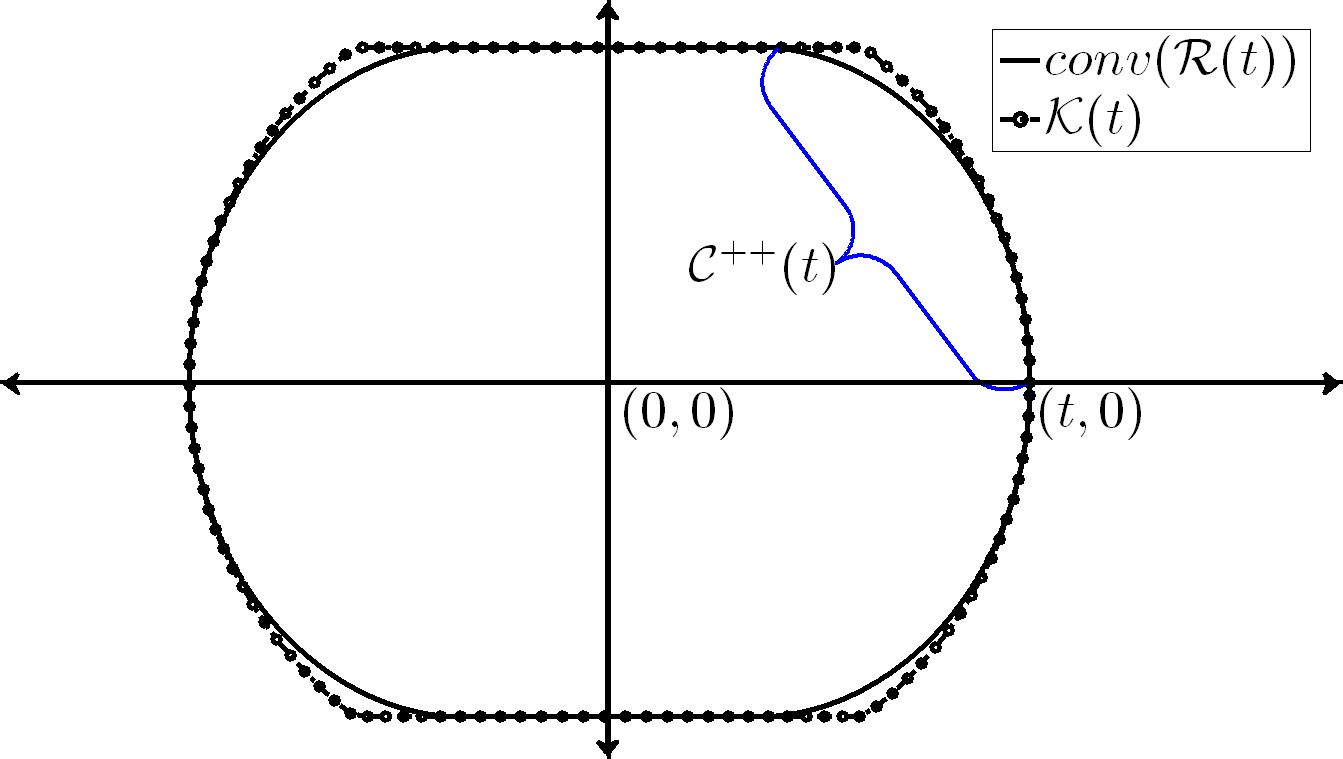}
		\caption{The set $conv(\mathcal{R}(t))$ is shown enclosed within $\mathcal{K}(t)$. In particular, it is illustrated how the curve $\mathcal{C}^{++}(t)$, which represents the curved boundary of $conv(\mathcal{R}(t))$ in the first quadrant, can be over-approximated by a circular arc of radius $t$, which forms the boundary of $\mathcal{K}(t)$ in the first quadrant.}
		\label{fig_approxRt}
	\end{figure}
	Also, the definition of $p_y(t,\psi)$ gives us the $y$-axis bounds of $conv(\mathcal{R}(t))$:
	\begin{align} \label{circ_approx_proof2}
	\max_{\psi \in [0,\min(t,\pi/2)]} p_y(t,\psi) = \begin{cases} 1-\cos(t), \;\; \text{if }  \; 0 < t \leq \pi/2  \\ t-\pi/2+1, \text{if } \; t > \pi/2 \end{cases}.
	\end{align}
	By symmetry, these results hold in all the 4 quadrants and it follows that $conv(\mathcal{R}(t)) \subset \mathcal{K}(t)$ (see Fig. \ref{fig_approxRt}). \par %(see Fig. \ref{fig_ellipseRt}). \par 
	Next, it is shown that, as $t$ grows larger, the dissimilarity between the two sets asymptotically goes to zero. Due to the asymptotic nature of the result, only values of time greater than $\pi/2$ are considered. Since $conv(\mathcal{R}(t)) \subset \mathcal{K}(t)$,  $\;conv(\mathcal{R}(t)) \cap \mathcal{K}(t) = conv(\mathcal{R}(t)) $ and  $\;conv(\mathcal{R}(t)) \cup \mathcal{K}(t) = \mathcal{K}(t)$.
	So, the Jaccard distance can be computed as,
	\begin{align} \label{eqn_jaccard_dist}
	d_J\big (conv(\mathcal{R}(t)),\mathcal{K}(t)\big ) = 1 - \frac{\mathbf{A}(conv(\mathcal{R}(t)))}{\mathbf{A}(\mathcal{K}(t))}.
	\end{align}
	Since both the sets are symmetric with respect to the $x$ and $y$ axes, area is computed in the first quadrant alone. After integrating over the boundary of $conv(\mathcal{R}(t))$, $\mathbf{A}(conv(\mathcal{R}(t)))$ is given by,
	\begin{equation} \label{eqn_area_final_convH}
	\mathbf{A}(conv(\mathcal{R}(t))) = \frac{1}{48}(12\pi(t^2-1) + t(48 - 6\pi^2) + \pi^3).
	\end{equation}
	Similarly, computing the area of $\mathcal{K}(t)$ in the first quadrant,
	\begin{multline} \label{delta_eqn_1}
	\mathbf{A}(\mathcal{K}(t)) = \frac{t^2}{2}\bigg (\sin^{-1}{(1-\delta(t))} + \\ 
	\frac{1}{2}\sin(2\sin^{-1}{(1-\delta(t))})\bigg ),
	\end{multline} 
	where $\delta(t) = (\pi/2 -1)/t$.
	As $t \rightarrow \infty$, the term $t^2$ in \eqref{eqn_area_final_convH} dominates and $\delta(t) \rightarrow 0$ in \eqref{delta_eqn_1}. So, for large values of $t$,
	\begin{equation} \label{eqn_kt_approx_area}
	\mathbf{A}(conv(\mathcal{R}(t))) \approx \frac{\pi}{4}t^2,\;\; \mathbf{A}(\mathcal{K}(t)) \approx \frac{\pi}{4}t^2.\\
	\end{equation}
	Evaluating $\lim_{t \to \infty} d_J(conv(\mathcal{R}(t)),\mathcal{K}(t))$ using \eqref{eqn_kt_approx_area}, the desired result is obtained.
\end{proof}
The set $\mathcal{K}(t)$ not only allows us to derive analytical expressions for the minimum area ellipse enclosing it, the asymptotic reduction in the dissimilarity between $\mathcal{K}(t)$ and $conv(\mathcal{R}(t))$ implies that, the impact of using $\xi(\mathcal{K}(t))$ instead of $\xi(conv(\mathcal{R}(t)))$ on the accuracy of set-membership tests in the safe time horizon algorithm, is minimal. \par
In $\mathbb{R}^n$, an ellipsoid can be represented uniquely by its center $c$ and a positive-definite matrix $H$: $E(c,H) = \{x \in \mathbb{R}^n: (x-c)^TH(x-c) \leq 1 \}$. According to results presented in \cite{john2014extremum}, at any given time $t$, the minimum area ellipse circumscribing $\mathcal{K}(t)$ can be obtained by solving the following semi-infinite programming problem:
\begin{align} \label{semi_inf_ellipse_1}
& \min_{c,H} - \log \: det(H) \\
& s.t. \;\; (z-c)^TH(z-c) \leq 1 , \forall z \in \mathcal{K}(t). \nonumber
\end{align} \par
The rest of this section formulates an analytical solution to this semi-infinite programming problem. In Lemma \ref{lemma_struct_ellipse}, we utilize the symmetry properties of $\mathcal{K}(t)$ to determine the center and orientation of the ellipse $\xi(\mathcal{K}(t))$. Following this, Lemma \ref{lemma_contact_points} and Lemma \ref{lemma_convex_program} pose the semi-infinite programming problem as a convex optimization problem. By solving this, we present analytical expressions for the ellipse $\xi(\mathcal{K}(t))$ in Theorem \ref{theorem_ellipse_exp}.
\begin{lem} \label{lemma_struct_ellipse}
	The ellipse $\xi(\mathcal{K}(t))$ has the form $E(c,H(t))$, where $c = (0,0)$ and $H(t) = diag(A(t),B(t))$ for some $A(t), B(t) \in \mathbb{R}$, such that $A(t) > 0 , B(t) > 0, \forall t > 0$.
\end{lem}
\begin{proof} 
	For a convex set $\mathbf{K} \subset \mathbb{R}^n$, denote $\mathcal{O}(\mathbf{K})$ as a set of affine transformations which leave the set $\mathbf{K}$ invariant:
	\begin{equation}
	\mathcal{O}(\mathbf{K}) = \{\mathbf{T}(x) = a + Px : \mathbf{T}(\mathbf{K}) = \mathbf{K}\}.
	\end{equation}
	This set is called the automorphism group of $\mathbf{K}$. Applying results from  \cite{guler2012symmetry}, we know that, $\mathcal{O}(\mathbf{K}) \subseteq \mathcal{O}(\xi(\mathbf{K}))$.
	Since $\mathcal{K}(t)$ is symmetric about both the $x$ and $y$ axes, the transformation $ \mathbf{T}(x) = -I_2x$, where $I_2$ is the identity matrix, lies in $\mathcal{O}(\mathcal{K}(t))$, and hence in $\mathcal{O}(\xi(\mathcal{K}(t)))$. Furthermore, if $\mathbf{T} \in \mathcal{O}(\xi(\mathcal{K}(t)))$ then $ \mathbf{T}(c) = c$. Applying this result with the transformation $\mathbf{T}(x) =  -I_2x$, we get $c = (0,0)$. \\
	We know from \cite{guler2012symmetry}, that $\mathbf{T} \in \mathcal{O}(\xi(\mathcal{K}(t))) \implies P^TH(t)P = H(t)$. The structure of $H(t)$ appears by applying
	\begin{equation}
	P = \begin{bmatrix} 1 & 0 \\ 0 & -1 \end{bmatrix},
	\end{equation}
	and the fact that $H(t)$ is always positive definite. 
\end{proof}
Applying results from \cite{john2014extremum} for the case of $\mathcal{K}(t) \subset \mathbb{R}^2$, we know that, there exist contact points $\{q_i\}_1^h$,  $ 0 < h \leq 5$  satisfying  $q_i \in \partial{\mathcal{K}(t)} \cap \partial{\xi(\mathcal{K}(t))}, i = 1,\ldots,h$. Furthermore, these contact points cannot all lie in any closed halfspace whose bounding hyperplane passes through the center of $\xi(\mathcal{K}(t))$. Using these facts, the following lemma can be stated:
\begin{lem} \label{lemma_contact_points}
	There are 4 contact points between $\mathcal{K}(t)$ and $\xi(\mathcal{K}(t))$.
\end{lem}
\begin{proof}
	We seek to find the number of contact points $h$. Since $\mathcal{K}(t)$ and $\xi(\mathcal{K}(t))$ are symmetric about the $x$ and $y$ axes (see Lemma \ref{lemma_struct_ellipse}), and $0 < h \leq 5$, $h$ can only take values 2 or 4. If $h=2$, both contact points must lie on the $x$ or $y$ axes (otherwise symmetry in all the 4 quadrants is not possible). But, the contact points cannot lie in any closed halfspace whose bounding hyperplane passes through the center of $\xi(\mathcal{K}(t))$. Hence $h=4$.
\end{proof}
Given the structure of $H(t)$ from Lemma \ref{lemma_struct_ellipse}, and the number of contact points from Lemma \ref{lemma_contact_points}, the semi-infinite programming problem \eqref{semi_inf_ellipse_1} will now be re-formulated as a convex optimization problem.
\begin{lem} \label{lemma_convex_program}
	The matrix $H(t) = diag(A(t),B(t))$ can be expressed as the solution of a convex optimization problem. At any given time $t > 0$, $A(t)$ is given as, 
	\begin{align} \label{semi_inf_ellipse_2}
	& A(t) = \argmin_{\mathbf{X}_t} \; - \log \; \mathbf{X}_t - \log \;\frac{ 1-\mathbf{X}_t(t^2-\alpha(t)^2)}{\alpha(t)^2}\\
	& s.t. \;\; 0 < \mathbf{X}_t \leq \frac{1}{t^2}. \nonumber
	\end{align}
	Furthermore,
	\begin{equation}
	B(t) = \frac{1 - A(t)(t^2-\alpha(t)^2)}{\alpha(t)^2}
	\end{equation}
	where $\alpha(t) = \begin{cases} 1-\cos(t), \;\; \text{ if } \; 0 < t \leq \pi/2 \\ t-\pi/2 + 1, \text{ if } \;\: t > \pi/2 \end{cases}$.
\end{lem}
\begin{proof} 
	As outlined in the semi-infinite programming problem given by \eqref{semi_inf_ellipse_1}, we aim to derive expressions for $A(t)$ and $B(t)$ which minimize the cost function,
	\begin{equation} \label{eqn_minimization_of_vol}
	 - \log\:det(H(t))  = -\log(A(t)) - \log(B(t)),
	 \end{equation}
	 subject to the constraint that $\mathcal{K}(t) \subset E((0,0),H(t))$. This constraint can be translated into the condition that the quadratic function,
	\begin{equation} \label{eqn_quad_function}
	g(y) = A(t)(t^2-y^2) + B(t)(y^2) - 1,
	\end{equation}
	is non-positive for $|y| \leq \alpha(t)$. \par	
	As outlined in Lemma \ref{lemma_contact_points}, there are a total of four contact points, one in each quadrant. Thus, there must be at least two distinct points at which $g(y)$ is zero. The roots of the quadratic function $g(y)$ must lie at $y = \pm \alpha(t)$ (in no other situation can $g(y)$ take non-positive values in the given interval). So $g(y)$ can be alternatively expressed as,
	\begin{equation} \label{eqn_quad_generic}
	g(y) = \lambda(y-\alpha(t))(y+\alpha(t)), \text{ for some } \lambda \geq 0.
	\end{equation}
	By equating the coefficients in \eqref{eqn_quad_function} and \eqref{eqn_quad_generic}, the following constraints emerge,  
	\begin{align*}
	& B(t) - A(t) = \lambda \geq 0, & A(t)t^2 - 1 + \lambda\alpha(t)^2 = 0.
	\end{align*}
	Eliminating $\lambda$ from the equations, the two constraints are reduced to $B(t) \geq A(t)$ and $ A(t)t^2 - 1 + (B(t)-A(t))\alpha(t)^2 = 0$. Expressing $B(t)$ in terms of $A(t)$ using the second constraint, we get, 
	\begin{align} \label{eqn_bt_into_at}
	A(t) \leq \frac{1}{t^2}, \; B(t) = \frac{1 - A(t)(t^2-\alpha(t)^2)}{\alpha(t)^2}. 
	\end{align}
	Substituting $B(t)$ from equation \eqref{eqn_bt_into_at} into the cost function \eqref{eqn_minimization_of_vol}, and denoting $\mathbf{X}_t$ as the value taken by the function $A(\cdot)$ at time $t$, we obtain the optimization problem given in \eqref{semi_inf_ellipse_2}, whose point-wise minimizer in time gives the value of $A(t)$.
\end{proof}

The next theorem solves the convex optimization problem outlined above to obtain analytical expressions for the minimum area ellipses enclosing $\mathcal{K}(t)$ (Fig. \ref{fig_ellipseRt}).

\begin{thm} \label{theorem_ellipse_exp}
	The minimum area ellipse $\xi(\mathcal{K}(t))$ has the form $E(c,H(t))$, where $c=(0,0)$ and $H(t) = diag(A(t),B(t))$. $A(t)$ and $B(t)$ are given by the following expressions: \\
	
	\begin{enumerate}
		\item If $0 < t \leq \pi/2$, then
		\begin{equation}
		A(t) = \frac{1}{2(t^2-\alpha(t)^2)} \text{  and  } B(t) = \frac{1}{2\alpha(t)^2}, 
		\end{equation}
		where $\alpha(t) = 1-\cos(t)$. 
		\item If $\pi/2 < t \leq (1+\frac{1}{\sqrt{2}})(\pi-2)$, then
		\begin{equation}
		A(t) = \frac{1}{2(t^2-\alpha(t)^2)} \text{  and  } B(t) = \frac{1}{2\alpha(t)^2}, 
		\end{equation}
		where $\alpha(t) = t-\pi/2+1$.
		\item If $t > (1+\frac{1}{\sqrt{2}})(\pi-2)$, 
		\begin{equation}
		A(t) = \frac{1}{t^2} \text{  and  } B(t) = \frac{1}{t^2} 
		\end{equation}	
	\end{enumerate}
\end{thm}
\begin{proof}
	We seek to calculate the minimizer to the convex cost function $f(\mathbf{X}_t)$ given in \eqref{semi_inf_ellipse_2}. In the interior of the constraint set, i.e., when 
	$\mathbf{X}_t < \frac{1}{t^2}$, the minimizer can be found by setting the gradient to zero and solving for $\mathbf{X}_t$:
	%We first search for optimal values of $A(t)$ in the interior of the constraint set, i.e. when $A(t) < \frac{1}{t^2}$. Denote the cost function in \eqref{semi_inf_ellipse_2} as $f(A(t))$.  The optimal value can be found by setting the gradient to zero and solving for $A(t)$:
	\begin{align*}
	%& \nabla{f(A)} = 0 \\
	& \nabla{f(\mathbf{X}_t)} = -\frac{1}{\mathbf{X}_t} + \frac{t^2 - \alpha(t)^2}{1-\mathbf{X}_t(t^2-\alpha^2)} = 0, \\
	& \mathbf{X}_t = \frac{1}{2(t^2-\alpha(t)^2)}
	\end{align*} \par
	For this expression to satisfy the constraint, the inequality $t^2 > 2\alpha(t)^2$ must be satisfied. This is true for all values of $t$ in the interval  $(0,\pi/2]$. For $t > \pi/2$, this holds whenever the function $t^2 + 4(1-\pi/2)t + 2(1-\pi/2)^2$ takes non-positive values. This is true for the time interval $0 < t \leq (1+\frac{1}{\sqrt{2}})(\pi-2)$. For all values of $t > (1+\frac{1}{\sqrt{2}})(\pi-2)$, the minimum in the feasible set is achieved when $\mathbf{X}_t = 1/t^2$. Evaluating $A(t)$ from \eqref{semi_inf_ellipse_2} and $B(t)$ from \eqref{eqn_bt_into_at}, we get the desired result.
\end{proof}
\begin{figure}[h]
	\centering
	\includegraphics[width=0.48\textwidth]{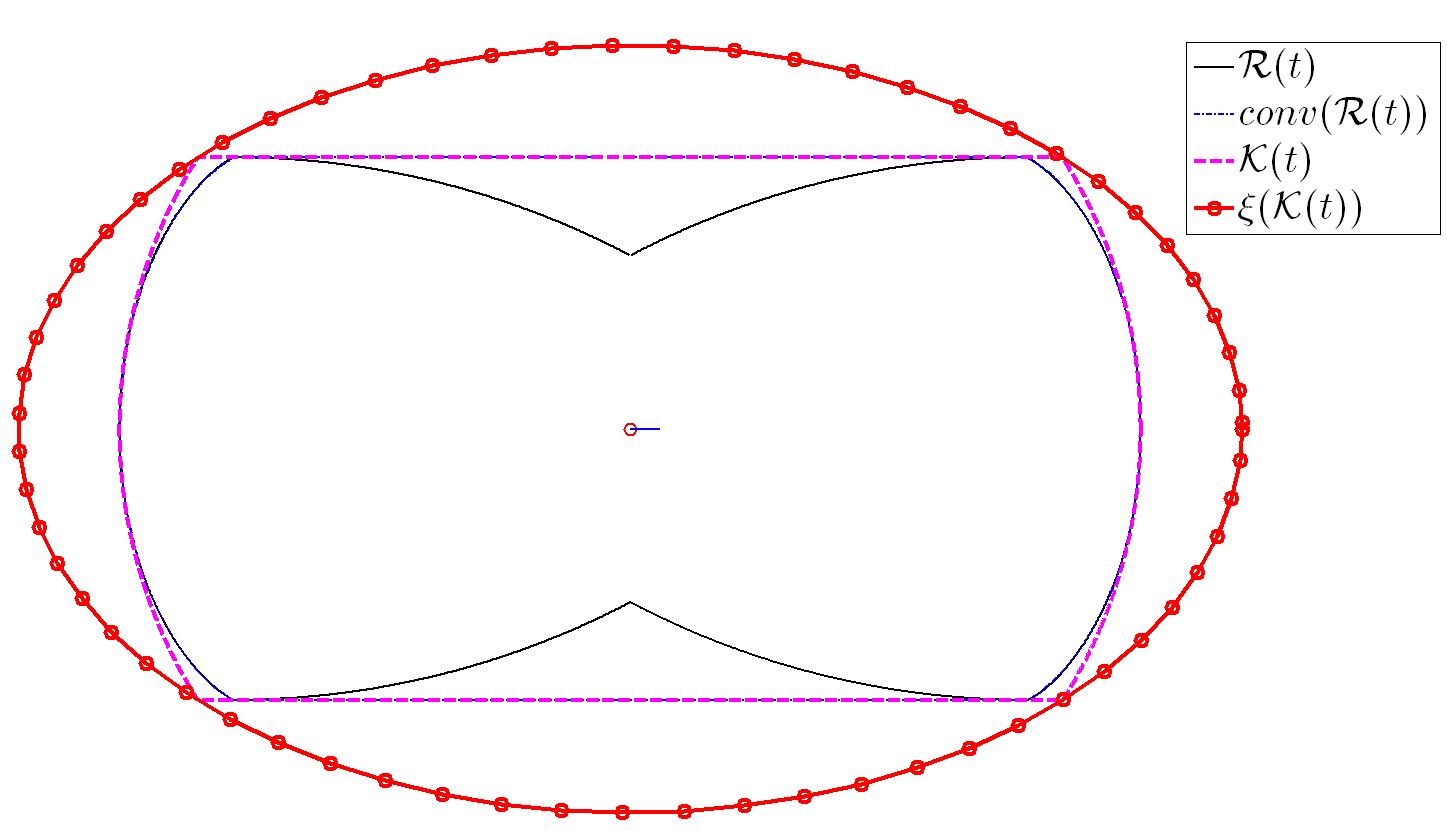}
	\caption{Minimum area ellipse enclosing $\mathcal{K}(t)$}
	\label{fig_ellipseRt}
\end{figure}
The following theorem justifies the ellipsoidal approximation by showing that the dissimilarity between $conv(\mathcal{R}(t))$ and $\xi(\mathcal{K}(t))$ asymptotically goes to zero as time grows larger.
\begin{thm}  \label{thm_convergence}
	Let $\mathcal{R}(t)$ denote the reachable set of a differential-drive robot with dynamics given in \eqref{unicycle_dynamics_normalized}, $conv(\mathcal{R}(t))$ denote its convex hull, $\mathcal{K}(t)$ denote an approximation of the convex hull as defined in Proposition \ref{prop_simple_R} and let $\xi(t)$ be the minimum area ellipse circumscribing $\mathcal{K}(t)$. Then,
	\begin{align*}
	\lim_{t \to \infty} d_J(conv(\mathcal{R}(t)),\xi(t)) = 0,
	\end{align*}
	where $d_J$ is the Jaccard distance. 
\end{thm}
\begin{proof}
	We first compute $d_J(\mathcal{K}(t),\xi(t))$. Since $\mathcal{K}(t) \subseteq \xi(t)$, $\mathcal{K}(t) \cup \xi(t) = \xi(t)$ and $\mathcal{K}(t) \cap \xi(t) = \mathcal{K}(t)$. So,
	\begin{equation}
	d_J(\mathcal{K}(t),\xi(t)) = 1 - \frac{\mathbf{A}(\mathcal{K}(t))}{\mathbf{A}(\xi(t))}.
	\end{equation} \par
	Due to the asymptotic nature of the result, only values of time greater than $(1+\frac{1}{\sqrt{2}})(\pi-2)$ are considered. For these values of time, $\xi(t)$ is a circle of radius $t$ (see Theorem \ref{theorem_ellipse_exp}). Also, since both $\xi(t)$ and $\mathcal{K}(t)$ are symmetric about the $x$ and $y$ axes, it suffices to compute areas in the first quadrant. So, $\mathbf{A}(\xi(t)) = \pi t^2/4$. From \eqref{eqn_kt_approx_area}, we know that, for large values of $t$, $ \mathbf{A}(\mathcal{K}(t)) \approx \pi t^2/4$
	which is equal to the area of $\xi(t)$ in the first quadrant. 
	Thus, 
	\begin{equation} \label{eqn_jaccard_k_e}
	\lim_{t \to \infty} d_J(\mathcal{K}(t),\xi(t)) = 0.
	\end{equation} 
	Since the Jaccard distance is a metric distance, it obeys the triangle inequality, 
	\begin{multline*}
	d_J(conv(\mathcal{R}(t)),\xi(t)) \leq d_J(conv(\mathcal{R}(t)),\mathcal{K}(t)) \\ + d_J(\mathcal{K}(t),\xi(t)).
	\end{multline*} \par
	Applying Proposition \ref{prop_simple_R} and \eqref{eqn_jaccard_k_e}, the desired result is obtained.
\end{proof}
Fig. \ref{fig_fourRt} shows the evolution of $\mathcal{R}(t), conv(\mathcal{R}(t)), \mathcal{K}(t)$ and $\xi(\mathcal{K}(t))$ for different values of time. As predicted by the result, the dissimilarity between $conv(\mathcal{R}(t))$, $\mathcal{K}(t)$ and $\xi(\mathcal{K}(t))$ asymptotically goes to zero. \par
\begin{figure}[h]
	\centering
	\includegraphics[width=0.49\textwidth]{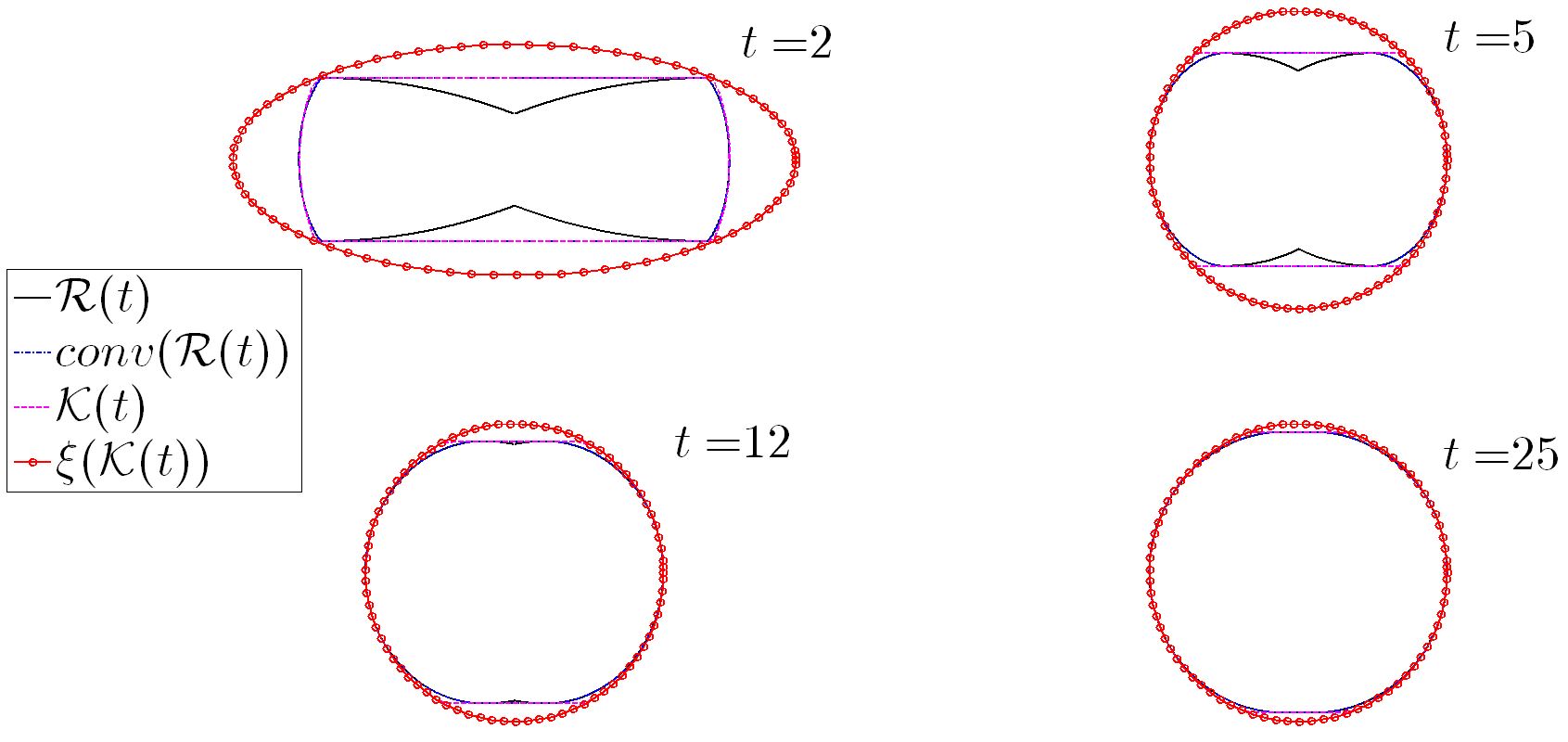}
	\caption{Evolution of $\mathcal{R}(t), conv(\mathcal{R}(t)), \mathcal{K}(t)$ and $\xi(\mathcal{K}(t))$ for $t = 2,5,12,25$. As predicted by Theorem \ref{thm_convergence}, the dissimilarity between $conv(\mathcal{R}(t)$ and $\xi(\mathcal{K}(t))$  asymptotically goes to zero. The scale for each figure is different.}
	\label{fig_fourRt}
\end{figure}
Using the ellipsoidal approximation of the reachable set developed in this section, Section \ref{sec_safe_time_horizon} outlines a safe time horizon based open-loop motion strategy for the robots.
\section{A Safe Open-Loop Motion Strategy} \label{sec_safe_time_horizon}
Let $\mathcal{M} = \{1,\ldots,N\}$ be a set of $N$ differential-drive robots, where each robot moves according to the dynamics specified in \eqref{unicycle_dynamics_normalized}. Let $(z_i(t),\phi_i(t))$ denote the configuration of robot $i \in \mathcal{M}$ at time $t$. At regular intervals of time $t_k = k\delta, k \in \mathbb{N}$, the central decision maker transmits the desired velocities $u_i(t_k) = (v_i(t_k),\omega_i(t_k))$ and the corresponding safe time horizon $s_i(t_k)$ to each robot $i \in \mathcal{M}$ via a wireless communication channel. $1/\delta$ is called the update frequency. \par
When a robot experiences communication failure, it executes the last received velocity command repeatedly for the duration of the corresponding safe time horizon. This causes the robot to follow a circular trajectory. Let $Z_i(\mu,t_k)$ denote the position of robot $i$ along this circular trajectory,
%\begin{equation}
%	Z_f(\mu,z_i(t_k)) = \begin{cases} 
where $t_k$ is the time of the last received command and $\mu$ is the time elapsed since the communication failure.
The expressions for $Z_i(\mu,t_k)$ are obtained by integrating \eqref{unicycle_dynamics_normalized} for constant velocity inputs. If $\omega_i(t_k) \neq 0$,
\begin{multline}
Z_i(\mu,t_k) = z_i(t_k) + \\ \frac{v_i(t_k)}{\omega_i(t_k)} \begin{pmatrix} \sin(\omega_i(t_k)\mu + \phi_i(t_k)) - \sin(\phi_i(t_k)) \\ \cos(\phi_i(t_k)) - \cos(\omega_i(t_k)\mu + \phi_i(t_k)) \end{pmatrix}
\end{multline}
and if $ \omega_i(t_k) = 0$,
\begin{equation}
Z_i(\mu,t_k) = z_i(t_k) + \mu v_i(t_k) \begin{pmatrix} \cos(\phi_i(t_k)) \\ \sin(\phi_i(t_k)) \end{pmatrix}.
\end{equation} \par
In order to ensure the scalability and computational tractability of the safe time horizon algorithm, we introduce the notion of a neighborhood set for each robot. To do this, the safe time horizon for each robot is upper-bounded by a pre-specified value $L$. This allows us to introduce the neighborhood set of robot $i$ at time $t$ as:
\begin{align} \label{eqn_neighbor_set}
N_i(t) = \big \{j \in \mathcal{M}, j \neq i: \|z_i(t) - z_j(t)\| < 2L \big \},
\end{align}
where $\|.\|$ denotes the $l_2$ norm. If robot $i$ and robot $j$ are not neighbors, they cannot collide within the maximum safe time horizon $L$. \par % 
The safe time horizon $s_i(t_k)$ can be defined as, 
\begin{equation}
s_i(t_k) = \min_{j \in N_i} s_{ij}(t_k),
\end{equation}
where $s_{ij}(t_k)$ is called the pair-wise safe time and is defined as,
\begin{align} \label{eqn_safe_time_def}
& s_{ij}(t_k) = \max_{\lambda} \int\limits_{0}^{\lambda} \; 1 \; d\lambda \\
\text{s.t. }Z_i(\mu,t_k) & \notin \mathcal{R}(\mu;z_j(t_k),\phi_j(t_k)), \forall \mu \in [0,\lambda] \nonumber \\
& \text{and } \lambda \leq L. \nonumber
\end{align}
%\end{defn}
Thus, the safe time horizon is the longest amount of time for which the trajectory of the robot after communication failure, does not intersect the reachable sets of its neighbors. But, as discussed in Section \ref{sec_approx_set}, an ellipsoidal approximation of the reachable set can be used to simplify set-membership tests. Thus, the definition of $s_{ij}(t_k)$ can be modified by replacing $\mathcal{R}(\mu;z_j(t_k),\phi_j(t_k))$ with $\xi_j(\mu,t_k)$ in \eqref{eqn_safe_time_def}, where $\xi_j(\mu,t_k)$ denotes the ellipsoidal approximation corresponding to $\mathcal{R}(\mu;z_j(t_k),\phi_j(t_k))$ as derived in Section \ref{sec_approx_set}. Next, we discuss how the safe time horizon is incorporated into the motion strategy of the robots. \par 
As discussed earlier, the central decision maker transmits $u_i(t_k)$ and $s_i(t_k)$ to all the robots $i \in \mathcal{M}$ at regular time intervals $t_k,\: k \in \mathbb{N}$. Let $c_i$ represent the status of the communication link of robot $i$:
\begin{align}
c_i(t_k) = \begin{cases}   1, \;\;\; \text{if } (u_i(t_k),s_i(t_k)) \text{ was received} \\ 
0, \;\;\; \text{if } (u_i(t_k),s_i(t_k)) \text{ was not received.} 
\end{cases} 
\end{align}
Algorithm \ref{alg_safe_time} outlines the motion strategy that robot $i$ employs. \par
\begin{algorithm}
	\caption{Safe Time Horizon based Open-Loop Motion Strategy}
	\label{alg_safe_time}
\begin{algorithmic}
	
	\State{$k = 1$, $l = 1$; $u_i(0) = 0, s_i(0) = 0$}
	\While{\texttt{true}}
	\If{$c_i(t_k) = 1$}
	\State{\texttt{Execute} $u_i(t_k)$}
	\State{$l=k$}
	\ElsIf {$t_k - t_l < s_i(t_{l})$}
	\State{\texttt{Execute} $u_i(t_{l})$} %\texttt{for} $s_i(t_{k-1})$ \texttt{seconds then Stop Moving}}
	\Else
	\State{\texttt{Stop Moving}} 
	\EndIf
	\State{$k = k + 1$}
	\EndWhile
\end{algorithmic}
\end{algorithm}
Fig. \ref{fig:safeTimeRationale} illustrates the rationale behind the safe time horizon algorithm. As long as the robot experiencing communication failure is outside the ellipsoidal reachable sets of its neighbors, it can safely move. The end of the safe time horizon corresponds to the time when the robot reaches the boundary of one of the ellipses. At this point, the robot stops moving.
\begin{figure}[h]
	\centering
	\framebox[0.46\textwidth]{\includegraphics[width=0.45\textwidth]{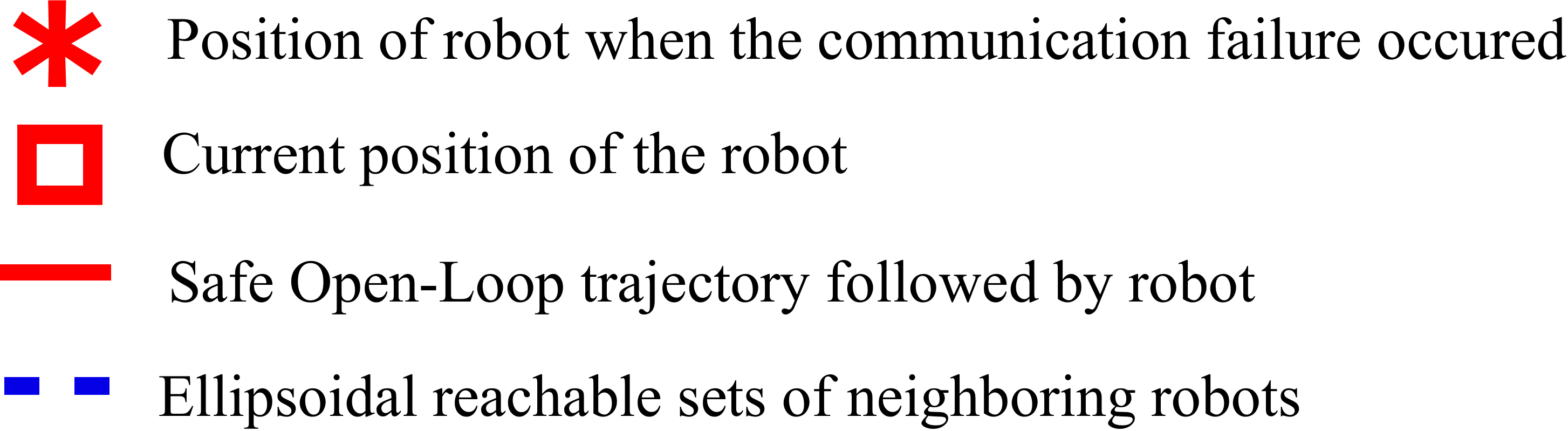}} \\
	\subfloat[][$t = 0.005$]{
	\includegraphics[width=0.22\textwidth,height = 0.18\textwidth]{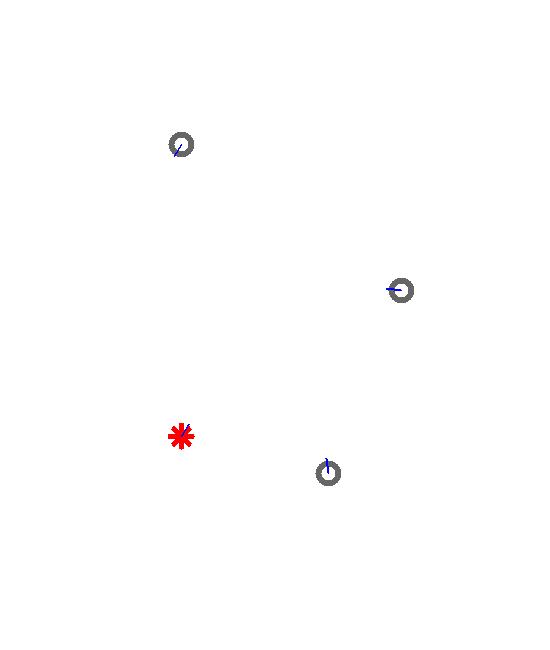}
	\label{subfig_a}}
 	\rulesep
	\subfloat[][$t = 3$]{
	\includegraphics[width=0.22\textwidth,height = 0.18\textwidth]{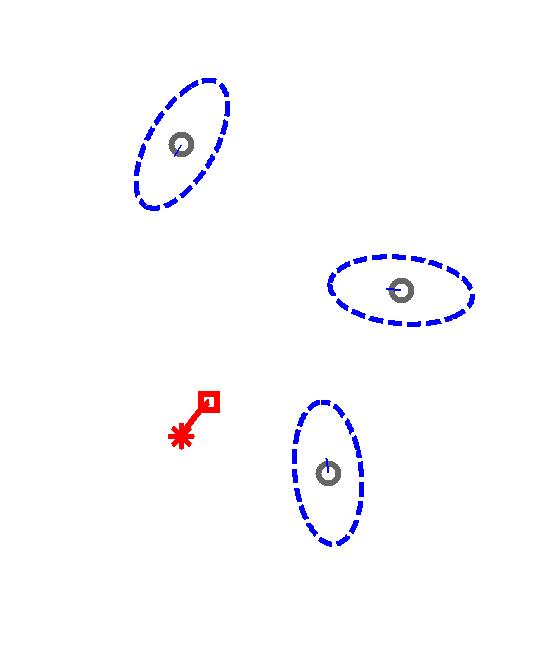}
	\label{subfig_b}}
	\rulehorzsep \\ 
	\subfloat[][$t = 5$]{
	\includegraphics[width=0.22\textwidth,height = 0.18\textwidth]{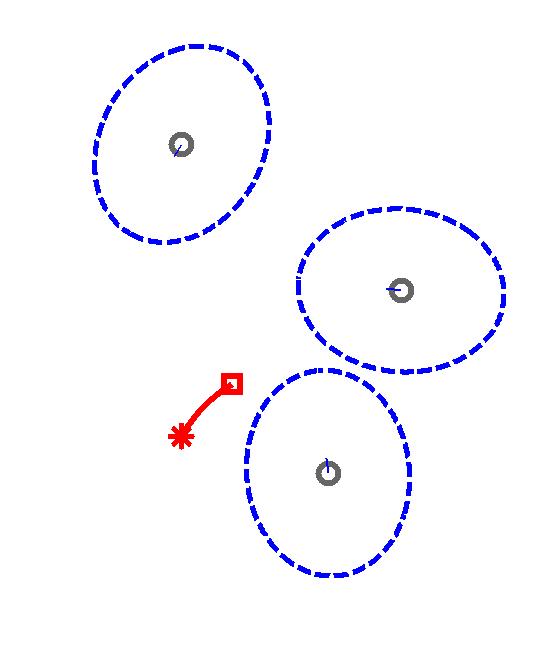}
	\label{subfig_c}}
	\rulesep
	\subfloat[][$t = 5.96$]{
	\includegraphics[width=0.22\textwidth,height = 0.18\textwidth]{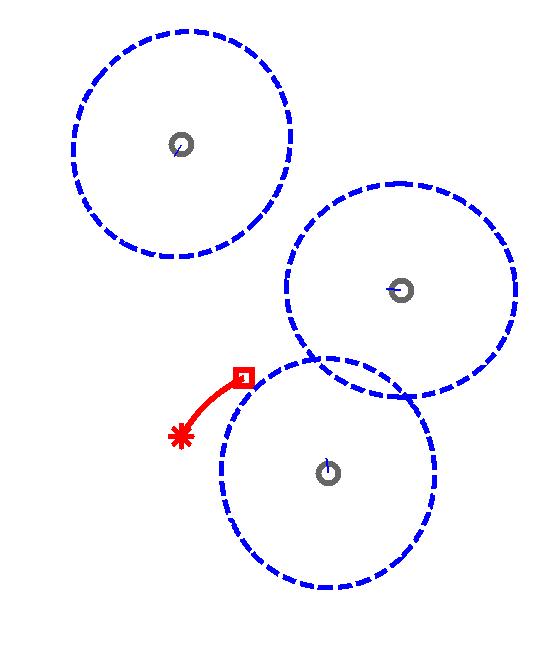}
	\label{subfig_d}}	
	\caption{The safe time horizon represents the longest time duration for which the robot lies outside the ellipsoidal reachable sets of other robots. Thus, the robot experiencing communication failure can execute its last received velocity command for the corresponding safe time horizon and remain safe. Beyond this, the robot stops moving.}
	\label{fig:safeTimeRationale}
\end{figure} \par
In order to state formal safety guarantees regarding Algorithm \ref{alg_safe_time}, we make mild assumptions on the capability of the control algorithm executing on the central decision maker. We assume that, the control algorithm ensures collision avoidance between communicating robots as well as between communicating robots and stationary obstacles. In particular, if $\exists i,j \in \mathcal{M}$ such that 
$c_i(t_k) = 1$ and $c_j(t_k) = 1$, then $u_i(t_k)$ and $u_j(t_k)$ guarantee that,
\begin{align}
\|z_i(t_k) - z_j(t_k) \| > 0 \implies \|z_i(t_{k+1}) - z_j(t_{k+1})\| > 0.
\end{align}
Let $z_O$ denote the position of a stationary obstacle. If $c_i(t_k) = 1$ for any $ i \in  \mathcal{M}$, 
\begin{align} \label{obs_control_guarantee}
\|z_i(t_k) - z_O \| > 0 \implies \|z_i(t_{k+1}) - z_O\| > 0.
\end{align}
Utilizing these assumptions, the following theorem outlines the safety guarantees provided by Algorithm \ref{alg_safe_time}. 
\begin{thm}
	If robot $i$ does not receive any commands from the central decision maker after time $t_k$, i.e., $c_i(t_k) = 1$ and $c_i(t_m) = 0$ $\forall m > \;k$, then Algorithm \ref{alg_safe_time} ensures that,  %$c_i(t_k) = 1$ and $c_i(t_m) = 0$ , $\forall m > \;k$, then Algorithm \ref{alg_safe_time} ensures that 
	\begin{align*}
	\|z_i(t_k) - z_j(t_k)\| > 0 & \implies \|z_i(t_k + \mu) - z_j(t_k + \mu)\| > 0 , \\
	& \forall \mu \in [0,s_i(t_k)], \; \forall j \in \mathcal{M}, j \neq i.
	\end{align*}
\end{thm}
\begin{proof}
	Since $\|z_i(t_k) - z_j(t_k)\| > 0 \; \forall j \in \mathcal{M}$, $s_i(t_k) > 0$. As seen in \eqref{eqn_neighbor_set}, if $j \notin N_i$, then a collision is not possible between robot $i$ and $j$ within the safe time horizon. Next, we consider the case when $j \in N_i$. From the definition of $s_i(t_k)$, we know that, $ z_i(t_k + \mu) \notin \xi_j(\mu,t_k), \; \forall \mu \in [0,s_i(t_k)]$. Furthermore, from the definition of reachable sets, $z_j(t_k + \mu) \in \xi_j(\mu,t_k), \; \forall \mu \in [0,s_i(t_k)]$. From the previous two statements, it is clear that $z_i(t_k + \mu) \neq z_j(t_k + \mu), \;\; \forall \mu \in [0,s_i(t_k)]$. Hence, 
	\begin{equation}
	\|z_i(t_k + \mu) - z_j(t_k + \mu)\| > 0, \;\;\; \forall \mu \in [0,s_i(t_k)], \; \forall j \in N_i.
	\end{equation}
	This completes the proof. 
\end{proof}
Beyond the safe time horizon, the robot stops moving, and \eqref{obs_control_guarantee} ensures that no collisions occur with the stationary robot. Thus, the original safety guarantee of the control algorithm is extended to situations where the robot is moving without commands from the central decision maker within the safe time horizon. \par
\section{Results} \label{sec_results}
\subsection{Simulations} \label{sec_sim}
This section presents the simulation results of the safe time horizon algorithm implemented on a team of 6 robots. Fig. \ref{fig_simJerkyMotion} compares the motion of the robots during a communication failure with and without the safe time horizon algorithm. A communication failure is simulated lasting from $t=3.1s$ to $t=8.3s$. In the case where safe time horizons are not utilized, shown by Fig. \ref{jerkyMotion_a} and Fig. \ref{jerkyMotion_b}, the robots experiencing communication failure abruptly stop moving, thus exhibiting a jerky motion pattern. When safe time horizons are utilized, the robots experiencing communication failure execute their last received velocity command for the duration of the safe time horizon (Fig. \ref{jerkyMotion_c} and Fig. \ref{jerkyMotion_d}). This allows them to keep moving during the communication failure, thereby avoiding jerky ``start-stop" motion behaviors and reducing the disruption caused to the multi-robot system.  
 \begin{figure}[h]
 	\centering
 	\framebox[0.46\textwidth]{\includegraphics[width=0.44\textwidth]{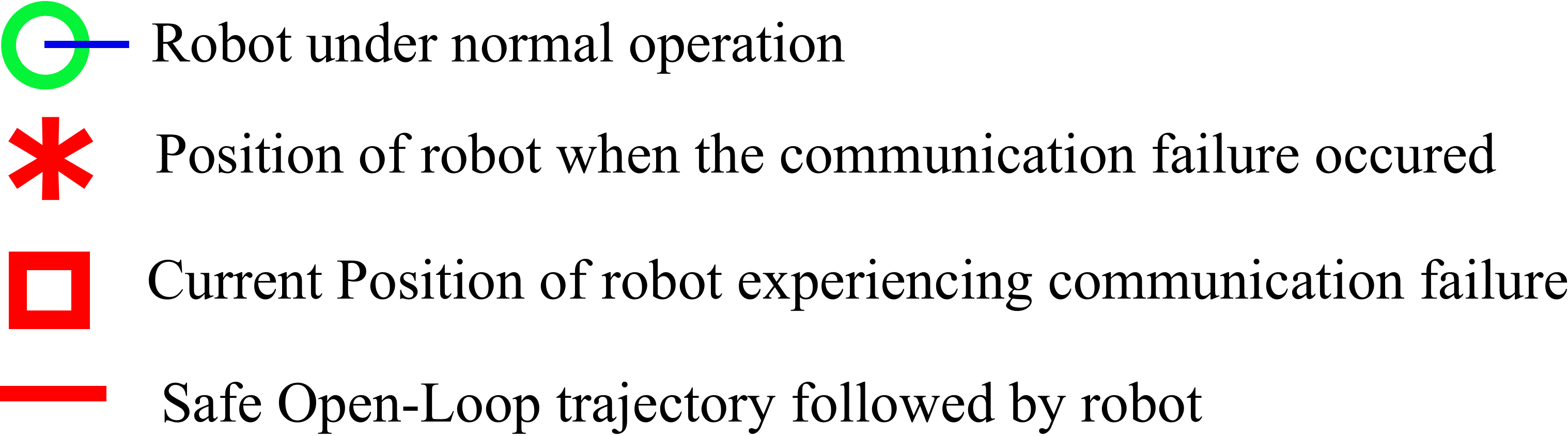}} \\
 	\subfloat[][No Safe Times: $t = 3.1s$]{
 		\includegraphics[width=0.22\textwidth,height=0.14\textwidth]{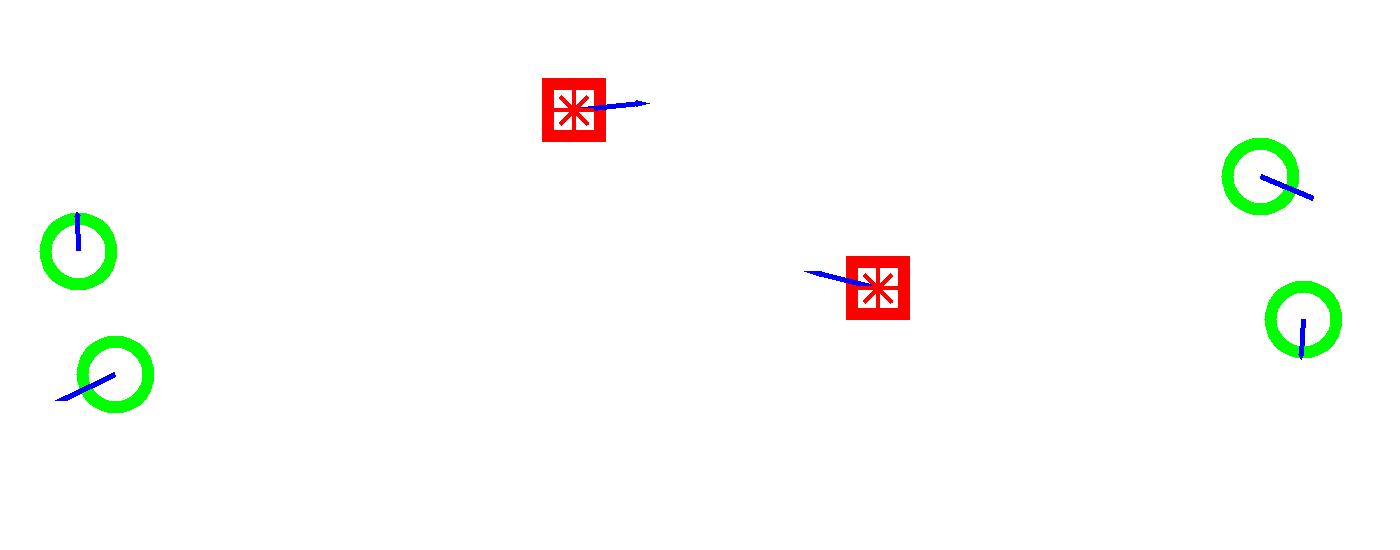}
 		\label{jerkyMotion_a}}
 	\rulesep
 	\subfloat[][No Safe Times: $t =8.3s$]{
 		\includegraphics[width=0.22\textwidth,height=0.14\textwidth]{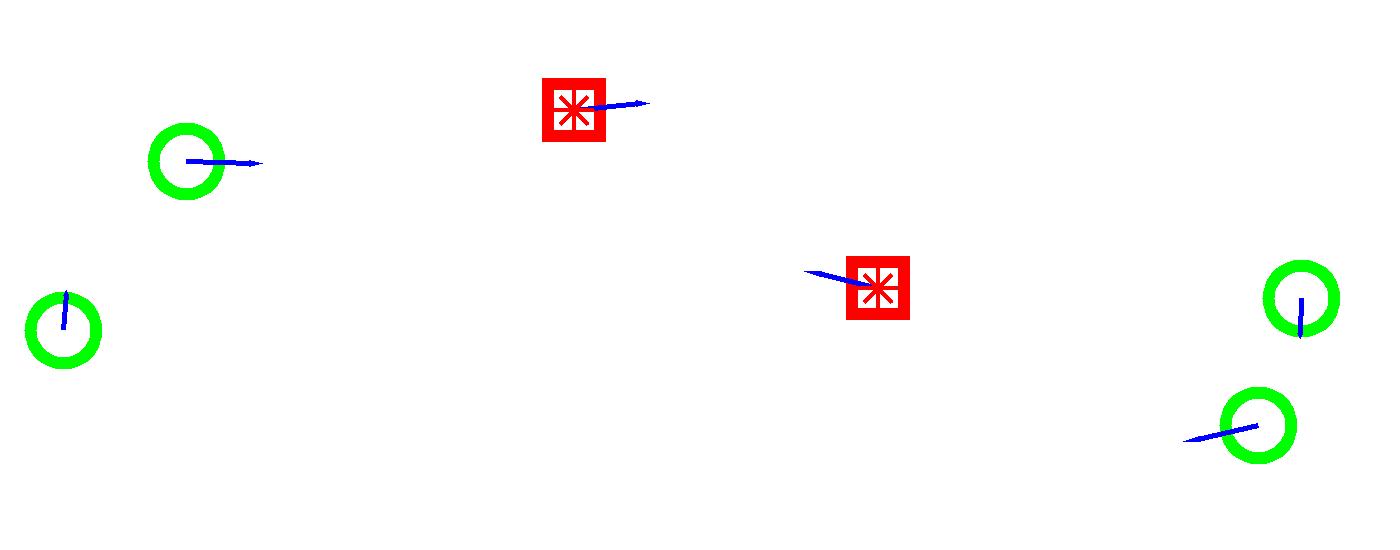}
 		\label{jerkyMotion_b}}
 	\rulehorzsep \\ 
 	\subfloat[][With Safe Times: $t = 3.1s$]{
 		\includegraphics[width=0.22\textwidth,height=0.14\textwidth]{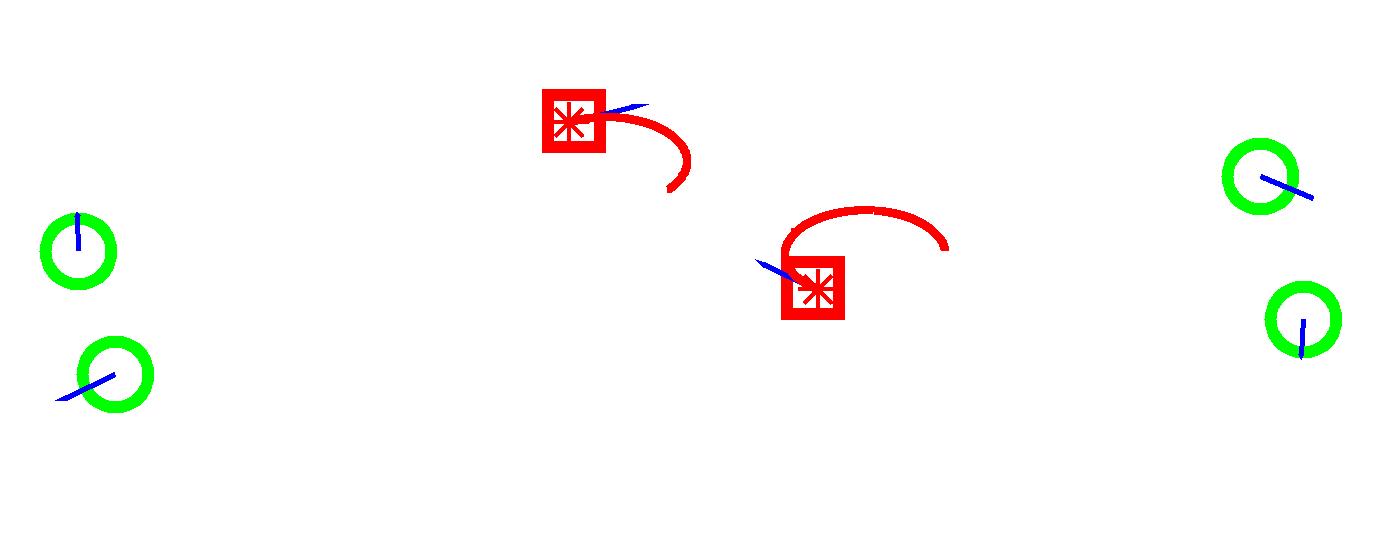}
 		\label{jerkyMotion_c}}
 	\rulesep
 	\subfloat[][With Safe Times: $t = 8.3s$]{
 		\includegraphics[width=0.22\textwidth,height=0.14\textwidth]{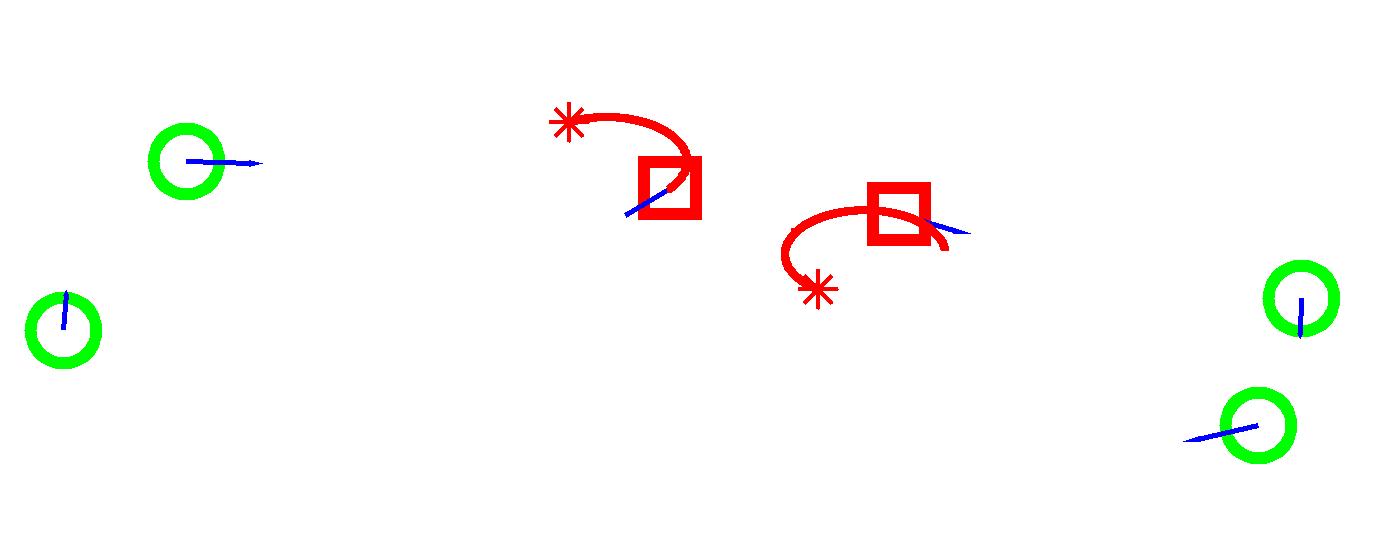}
 		\label{jerkyMotion_d}}
 	\caption{Comparison of the motion of robots with and without safe time horizons. Two robots experience communication failure from $t=3.1s$ to $t=8.3s$. In the case when safe time horizons are not used (Fig. \ref{jerkyMotion_a} and Fig. \ref{jerkyMotion_b}), the robots exhibit jerky motion behavior, since they abruptly stop during the communication failure. When safe time horizons are used (Fig. \ref{jerkyMotion_c} and Fig. \ref{jerkyMotion_d}), the robots continue moving by executing their last received velocity command for the corresponding safe time horizon, thus demonstrating the ability of the safe time horizon algorithm to effectively handle communication failures.}
 	\label{fig_simJerkyMotion}
 \end{figure}
\subsection{Experimental Results}
 The safe time horizon algorithm is implemented on a multi-robot testbed with 4 Khepera III robots and an Optitrack motion capture system, which is connected to a desktop computer serving as the central decision maker. In the first scenario (Fig. \ref{scene1_a} and Fig. \ref{scene1_b}), the robot experiencing communication failure executes its last received velocity command for the duration of the safe time horizon, after which it becomes stationary. In the second scenario, (Fig. \ref{scene1_c} and Fig. \ref{scene1_d}), the robot experiences a short duration communication failure and keeps moving safely through it. In particular, the safe time horizon algorithm successfully combats issues pertaining to intermittent communications on the Robotarium, and enables the seamless execution of coordination algorithms in the face of such failures.
 	\begin{figure}[h]
 		\centering
 		\framebox[0.46\textwidth]{\includegraphics[width=0.44\textwidth]{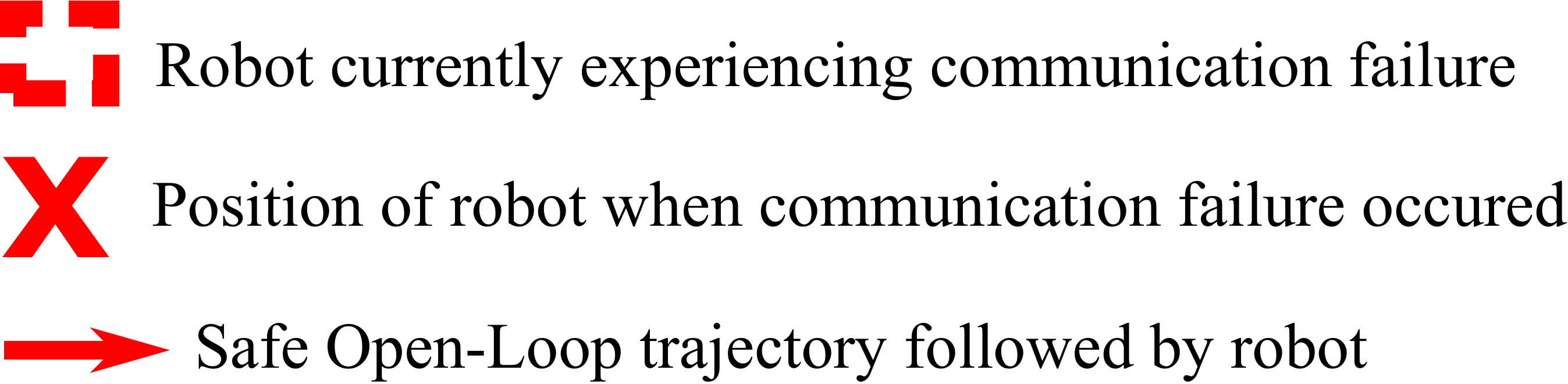}} \\
 		\subfloat[][Robots at $t = 9s$]{
 			\includegraphics[width=0.225\textwidth,height=0.15\textwidth]{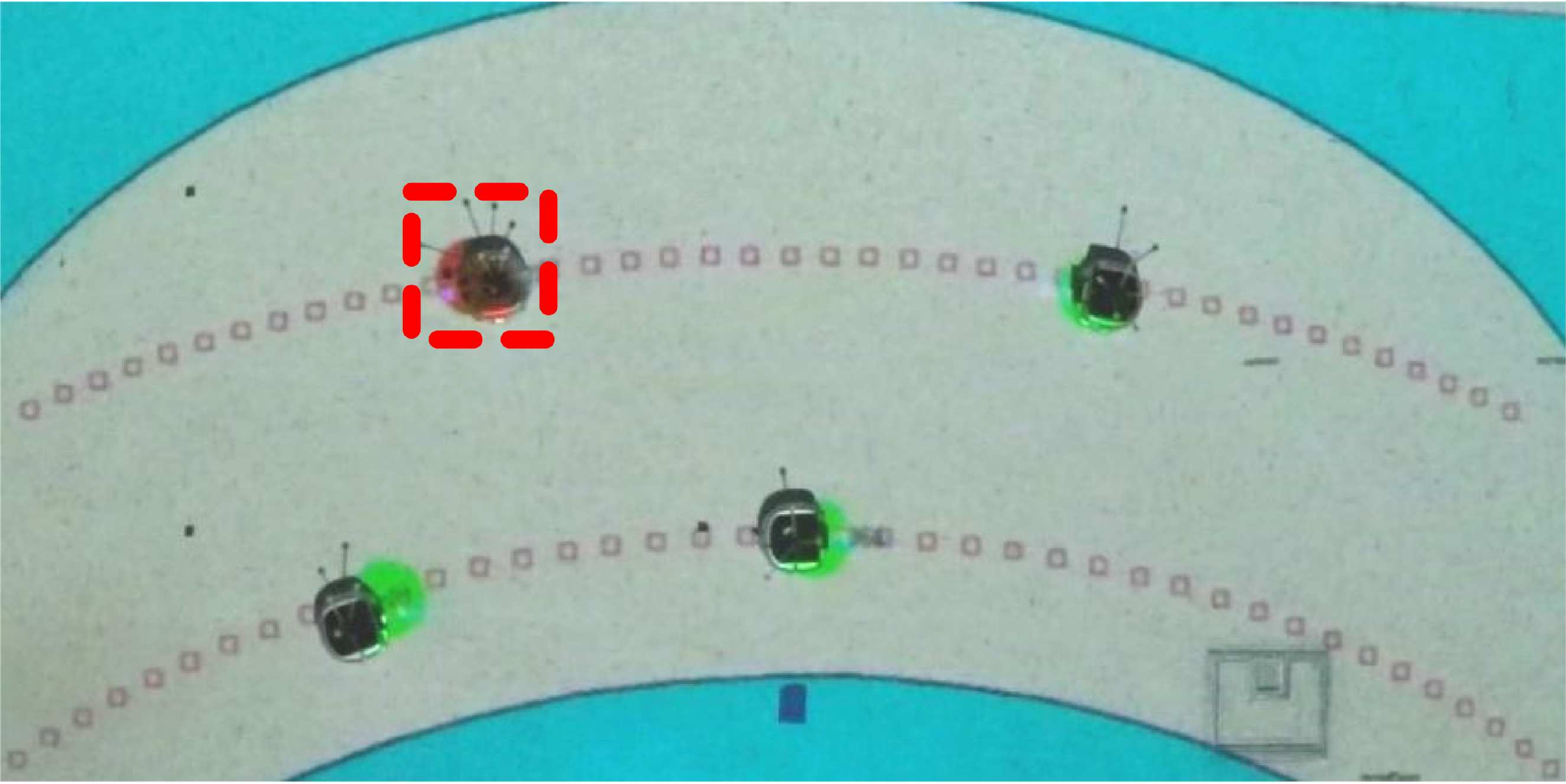}
 			\label{scene1_a}}
 		\subfloat[][Robots at $t = 14s$]{
 			\includegraphics[width=0.225\textwidth,height=0.15\textwidth]{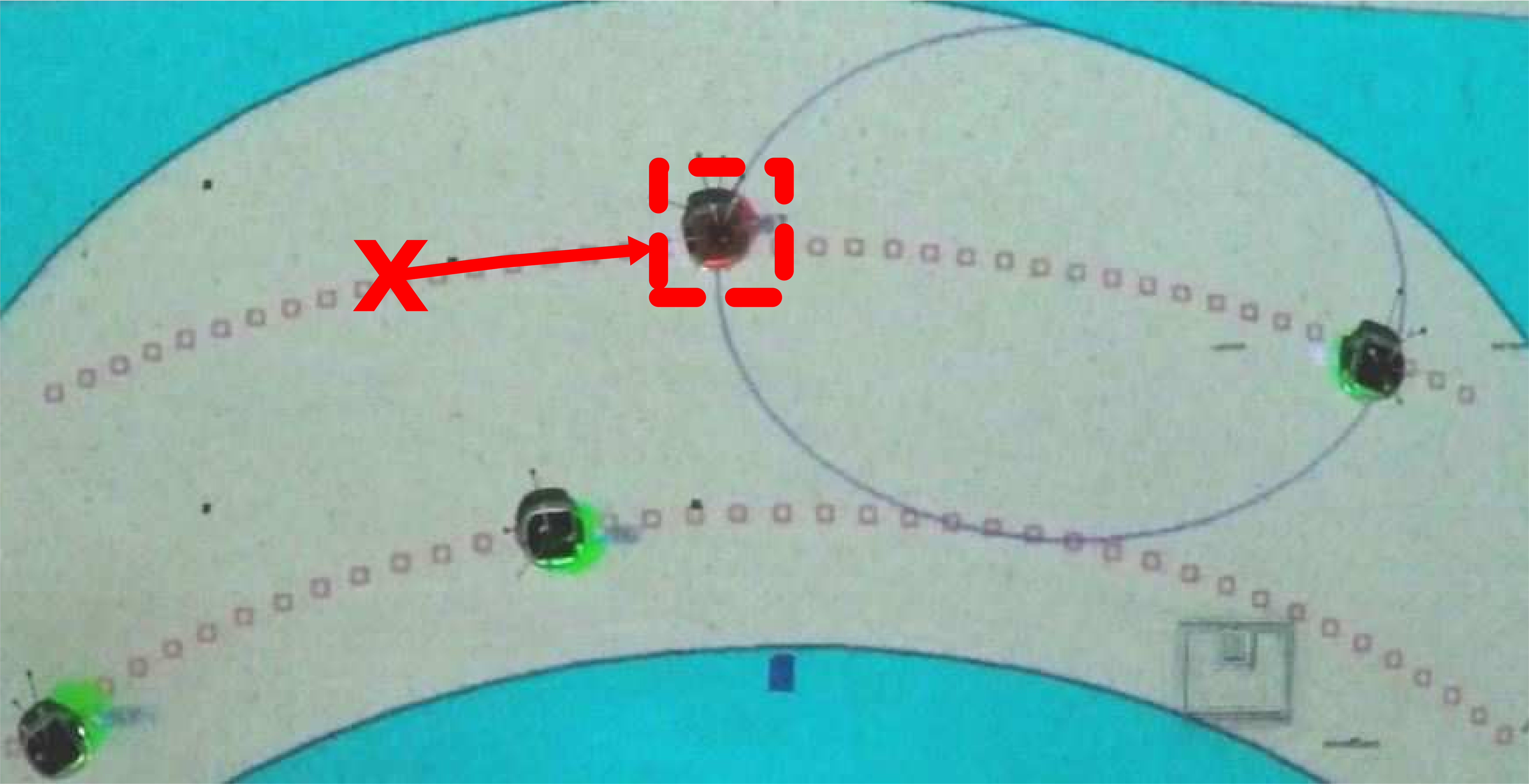}
 			\label{scene1_b}}
 		\\
 		\subfloat[][Robots at $t = 28s$]{
 			\includegraphics[width=0.225\textwidth,height=0.15\textwidth]{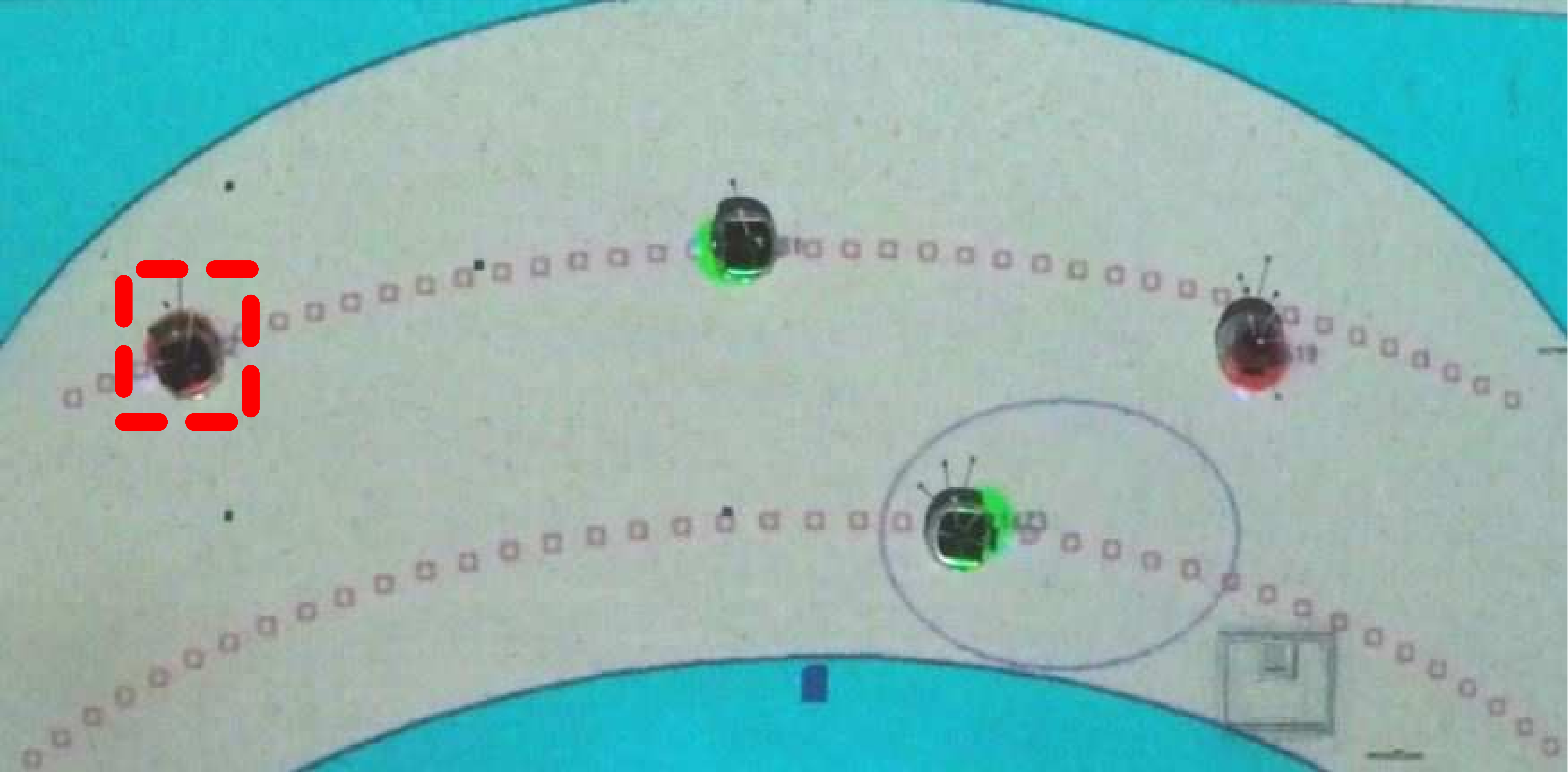}
 			\label{scene1_c}}
 		\subfloat[][Robots at $t = 32s$]{
 			\includegraphics[width=0.225\textwidth,height=0.15\textwidth]{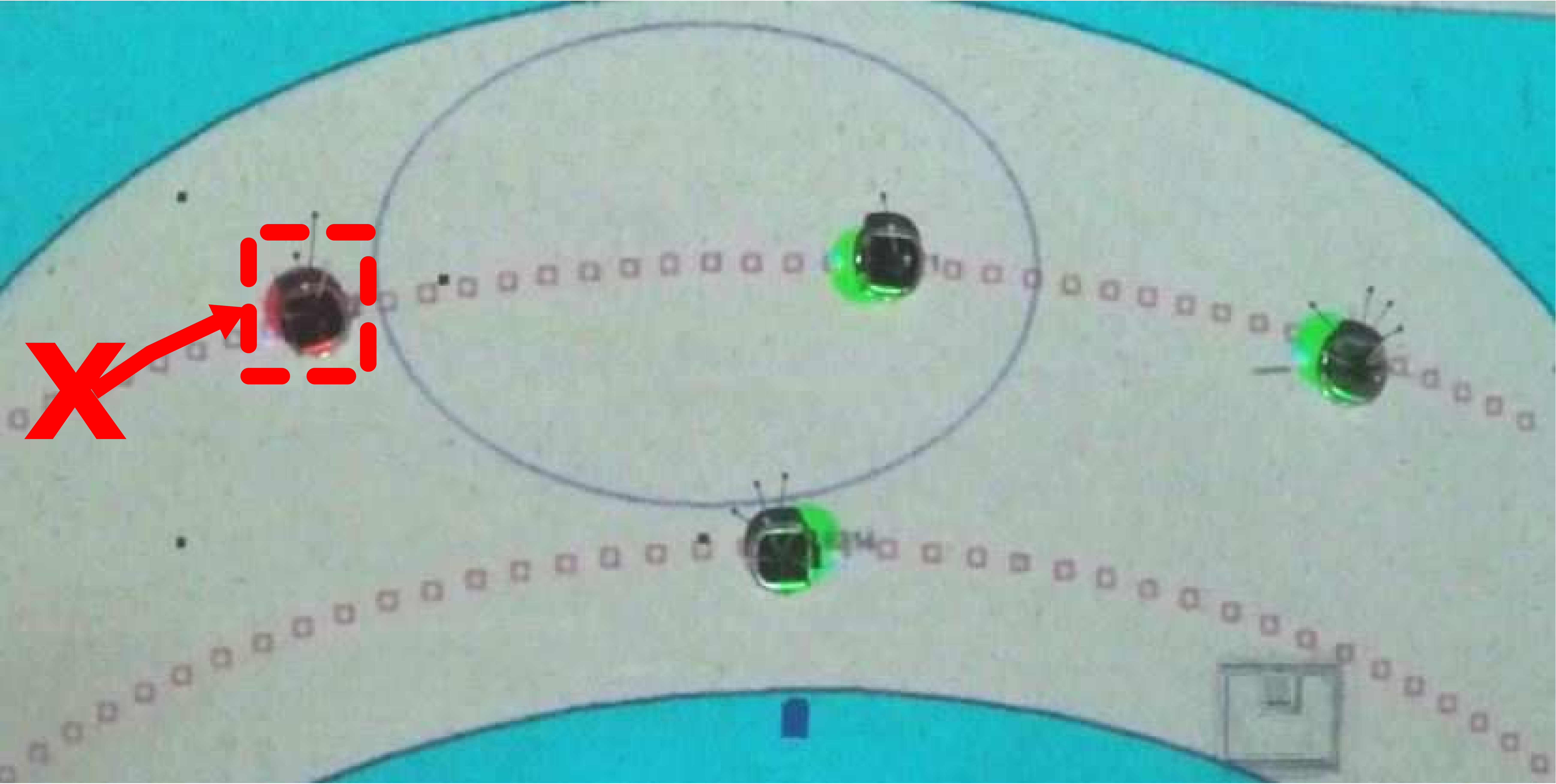}
 			\label{scene1_d}}
 		\caption{Four Khepera III robots are shown patrolling a U-shaped corridor. At $t=9s$, a robot experiences communication failure and executes its last received velocity command until $t=14s$. Similarly, a robot loses communication at $t=28s$ (Fig. \ref{scene1_c}) and keeps moving in a safe manner. Before the safe time horizon of the robot elapses, communication is restored (Fig. \ref{scene1_d}). Thus, the safe time horizon algorithm prevents jerky ``start-stop" motion patterns of the robot during the intermittent communication failure. This experiment demonstrates how the safe time horizon algorithm can prevent disruptions in robot motion caused due to intermittent communications on the Robotarium. A video of this experiment can be found at \texttt{www.youtube.com/watch?v=Gyz861xwaHY}}
 		\label{fig_scenario1}
 	\end{figure}
\section{Conclusions} \label{sec_conclusion}
The safe time horizon algorithm not only provides a technique for multi-robot systems to safely handle intermittent communication failures, it demonstrates the feasibility of reachability analysis as a powerful tool for multi-robot algorithms. The minimum area ellipse derived in Section \ref{sec_approx_set} provides a compact and efficient way to represent the reachable set of a differential-drive robot and can be used in other robotics algorithms as an abstraction of the reachable set itself.
%%%%%%%%%%%%%%%%%%%%%%%%%%%%%%%%%%%%%%%%%%%%%%%%%%%%%%%%%%%%%%%%%%%%%%%%%%%%%%%%
\bibliography{references}
\bibliographystyle{plain}
%%%%%%%%%%%%%%%%%%%%%%%%%%%%%%%%%%%%%%%%%%%%%%%%%%%%%%%%%%%%%%%%%%%%%%%%%%%%%%%%
\end{document}